\documentclass[12pt,a4paper]{article}

\usepackage{amsmath,amssymb}
\usepackage{mathtools}
\usepackage{graphicx}
\usepackage[nosort]{cite}

\usepackage{titlesec}
\usepackage{bold-extra}
\titleformat*{\section}{\large\bfseries\scshape}
\titleformat*{\subsection}{\large\scshape}
\titleformat*{\paragraph}{\normalsize\bfseries}

\usepackage{tikz}
\usetikzlibrary{decorations.markings}
\usepackage{epsfig}  
\usepackage{enumitem}
\usepackage{relsize}
\usepackage{dsfont}
\usepackage{stmaryrd}
\usetikzlibrary{calc} 
\usetikzlibrary{patterns} 

\newcommand{\mathsym}[1]{{}}

\usepackage{amsthm,amsbsy,color,multicol}
\usepackage{array,calc}
\usepackage{rotating}
\usepackage{a4wide,bm}
\usepackage{bbm}
\usepackage[a4paper,text={450pt,650pt},centering]{geometry}
\usepackage{setspace}
\usepackage{marginnote}

\let\oldbfseries=\bfseries
\let\oldmdseries=\mdseries
\let\oldnormalfont=\normalfont
\renewcommand{\bfseries}{\oldbfseries\boldmath}
\renewcommand{\mdseries}{\oldmdseries\unboldmath}
\renewcommand{\normalfont}{\oldnormalfont\unboldmath}

\allowdisplaybreaks[3]

\numberwithin{equation}{section}

\makeatletter
\renewcommand\subparagraph{\@startsection{subparagraph}{5}%
	{\parindent}
	{0pt}
	{-1em}
	{\normalfont\itshape}
} 
\makeatother  

\usepackage[font=small,labelfont=bf,width=0.85\textwidth]{caption}

\ifx\hypersetup\sadfkjashdfkxja\newcommand\hypersetup[1]{}\fi

\hypersetup{plainpages=false}
\hypersetup{pdfpagemode=UseNone}
\hypersetup{bookmarksnumbered=true}
\hypersetup{pdfstartview=FitH}
\hypersetup{colorlinks=false}
\hypersetup{citebordercolor={.5 1 .5}}
\hypersetup{urlbordercolor={.5 1 1}}
\hypersetup{linkbordercolor={1 .7 .7}}

\DeclareMathSymbol{\Gamma}{\mathalpha}{letters}{"00}
\DeclareMathSymbol{\Delta}{\mathalpha}{letters}{"01}
\DeclareMathSymbol{\Theta}{\mathalpha}{letters}{"02}
\DeclareMathSymbol{\Lambda}{\mathalpha}{letters}{"03}
\DeclareMathSymbol{\Xi}{\mathalpha}{letters}{"04}
\DeclareMathSymbol{\Pi}{\mathalpha}{letters}{"05}
\DeclareMathSymbol{\Sigma}{\mathalpha}{letters}{"06}
\DeclareMathSymbol{\Upsilon}{\mathalpha}{letters}{"07}
\DeclareMathSymbol{\Phi}{\mathalpha}{letters}{"08}
\DeclareMathSymbol{\Psi}{\mathalpha}{letters}{"09}
\DeclareMathSymbol{\Omega}{\mathalpha}{letters}{"0A}

\newcommand{\gen}[1]{\mathrm{#1}}

\newcommand{\dd}{\mathrm{d}}

\newcommand*\widebar[1]{%
  \hbox{%
    \vbox{%
      \hrule height 0.5pt 
      \kern0.25ex
      \hbox{%
        \kern-0.3em
        \ensuremath{#1}%
        \kern-0.1em
      }%
    }%
  }%
}

\ifx\genfrac\sdflkaj\else\fi

\newcommand{\beq}{\begin{equation}}
\newcommand{\eeq}{\end{equation}}

\def\[{\begin{equation}}
\def\]{\end{equation}}
\def\<{\begin{eqnarray}}
\def\>{\end{eqnarray}}

\newtheorem{mydef}{Definition}

\newtheorem{theorem}{Theorem}
 
\newtheorem{remark}{Remark}

\newtheorem{corollary}{Corollary}

\makeatletter
\def\mr@ignsp#1 {\ifx\:#1\@empty\else #1\expandafter\mr@ignsp\fi}%
\newcommand{\multiref}[1]{\begingroup
\xdef\mr@no@sparg{\expandafter\mr@ignsp#1 \: }%
\def\mr@comma{}%
\@for\mr@refs:=\mr@no@sparg\do{\mr@comma\def\mr@comma{,}\ref{\mr@refs}}%
\endgroup}
\makeatother

\newcommand{\hypref}[2]{\ifx\href\asklfhas #2\else\href{#1}{#2}\fi}
\newcommand{\Secref}[1]{Section~\multiref{#1}}

\renewcommand{\eqref}[1]{(\multiref{#1})}

\makeatletter
\newlength{\apb@width}
\newcommand{\autoparbox}[2][c]{\settowidth{\apb@width}{#2}\parbox[#1]{\apb@width}{#2}}

\makeatother

\ifx\href\asklfhas\newcommand{\href}[2]{#2}\fi

\begin{document}

\renewcommand{\thefootnote}{\fnsymbol{footnote}}
\thispagestyle{empty}
\begin{flushright}\footnotesize
ZMP-HH/15-21
\end{flushright}
\vspace{2cm}

\begin{center}%
{\Large\bfseries%
\hypersetup{pdftitle={First-order PDEs and the six-vertex model}}%
New differential equations \\ in the six-vertex model%
\par} \vspace{2cm}%

\textsc{W.~Galleas}\vspace{5mm}%
\hypersetup{pdfauthor={Wellington Galleas}}%

\textit{II. Institut f\"ur Theoretische Physik \\ Universit\"at Hamburg, Luruper Chaussee 149 \\ 22761 Hamburg, Germany}\vspace{3mm}%

\verb+wellington.galleas@desy.de+ %

\par\vspace{2.5cm}

\textbf{Abstract}\vspace{7mm}

\begin{minipage}{12.7cm}
This letter is concerned with the analysis of the six-vertex model with domain-wall boundaries in terms
of partial differential equations (PDEs). The model's partition function is shown to obey a system of PDEs
resembling the celebrated Knizhnik-Zamolodchikov equation. The analysis of our PDEs naturally produces a
family of novel determinant representations for the model's partition function.

\hypersetup{pdfkeywords={Six-vertex model, Knizhnik-Zamolodchikov equation, domain-wall boundaries }}%
\hypersetup{pdfsubject={}}%

\end{minipage}
\vskip 0.1cm
{\small Keywords: Partial differential equations, six-vertex model, \\ domain-wall boundaries}
\vskip 1.5cm
{\small November 2015}

\end{center}

\newpage
\renewcommand{\thefootnote}{\arabic{footnote}}
\setcounter{footnote}{0}

\tableofcontents

\bigskip
\section{Introduction}
\label{sec:INTRO}

The formulation of the Quantum Inverse Scattering Method (QISM) \cite{Sk_Faddeev_1979,Takh_Faddeev_1979} represented a large step towards
the understanding  of algebraic structures underlying exactly solvable models of Statistical Mechanics \cite{Baxter_book}.
In particular, it benefited from many insights gained through the study of the six-vertex model originally proposed to explain
the ice residual entropy \cite{Pauling_1935}. The relevance of two-dimensional vertex models to other fields has gradually
increased over the years and nowadays we can find ramifications from condensed matter physics to gauge theories. For instance, the relation
between two-dimensional classical vertex models and one-dimensional quantum spin chains \cite{Baxter_1971} is the most notorious one but
recent applications also include the study of supersymmetric gauge theories \cite{Nekrasov_2009}. In the latter case exactly solvable vertex
models are among the main players in the so-called Bethe/gauge correspondence describing the supersymmetric vacua of two-dimensional $\mathcal{N}=2$ gauge theories in terms 
of Bethe vectors \cite{Nekrasov_2009}. In addition to that, they also play a prominent role in the program initiated in \cite{Zarembo_2003} enabling  
integrability based techniques for the computation of anomalous dimensions of certain gauge invariant operators in $\mathcal{N}=4$ super 
Yang-Mills theory.

Roughly speaking, the formulation of a vertex model involves three main ingredients: graphs on a lattice, statistical weights for graph configurations
and boundary conditions. Although those ingredients are arbitrary a priori, integrability in the sense of Baxter \cite{Baxter_book, Baxter_1971} imposes 
restrictions on all of them. The statistical weights are then required to satisfy the Yang-Baxter equation while boundary conditions are
usually constrained by compatibility conditions \cite{deVega_1984, Sklyanin_1988}. The latter often consist of algebraic
equations but some types of integrable boundary conditions are not characterized in such a way. This is precisely the case
for domain-wall boundary conditions introduced in \cite{Korepin_1982} as a building block of scalar products of Bethe vectors. 

The six-vertex model with domain-wall boundaries has also found several important applications ranging from enumerative combinatorics
to the study of gauge theories. For instance, in \cite{Kuperberg_1995} Kuperberg demonstrated that the partition function of 
the six-vertex model with domain-wall boundaries exhibits a combinatorial line counting the number of Alternating Sign Matrices \cite{Mills_1983}.
In the gauge theory front, the chiral partition function of a two-dimensional Yang-Mills theory with gauge group $SU(N)$ was shown 
to correspond to the partition function of the six-vertex model with domain-wall boundaries in the ferroelectric regime \cite{Szabo_2012}. 
The literature devoted to this particular variant of the six-vertex model is quite extensive and we refer the reader to
\cite{Korepin_1982, Izergin_1987, Lascoux_2007, Bleher_2006, Bleher_2009, Bleher_2010} and references therein for a detailed account.
In particular, its partition function can be expressed as determinants \cite{Izergin_1987, Colomo_Pronko_2008} or alternatively as a multiple contour integral
\cite{deGier_Galleas_2011, Galleas_2012, Galleas_2013}. As far as the integral representation is concerned, it is worth remarking that it
was obtained in \cite{Galleas_2012, Galleas_2013} as a solution of certain functional equations governing the model's partition function. 
Those functional relations have their roots in the Yang-Baxter algebra and, interestingly, they also give rise to a hierarchy of
higher-order PDEs satisfied by the six-vertex model partition function \cite{Galleas_2011, Galleas_proc}. 

The analysis of this variant of the six-vertex model in terms of PDEs is the main purpose of this paper. However, the PDEs we shall
consider here are not the ones derived in \cite{Galleas_proc}; although they originate from the very same functional relation obtained in
\cite{Galleas_2013}. One important feature of these new PDEs is that they consist of linear first-order equations resembling
the celebrated Knizhnik-Zamolodchikov equation \cite{Knizhnik_1984}. Moreover, we shall demonstrate here how one can simply read off
the solution from the structure of the equations. Through this analysis we obtain a whole family of determinant representations 
as a result of Cramer's method for solving linear systems.

\paragraph{Outline.} This paper is organized as follows. In \Secref{sec:FUN} we describe certain functional equations satisfied by
the partition function of the six-vertex model with domain-wall boundaries. Those equations have their roots in the Yang-Baxter algebra
and in \Secref{sec:FUN} we also discuss some of their properties which will be required throughout this paper. 
In \Secref{sec:KZ} we describe a procedure producing first-order PDEs from the aforementioned functional equations. 
Through this approach we obtain PDEs with structure similar to that of Knizhnik-Zamolodchikov equations. The resolution of our
PDEs is then discussed in \Secref{sec:SOL} while concluding remarks are left for \Secref{sec:CONCLUSION}.

\section{Functional relations}
\label{sec:FUN}

Several methods have been formulated along the years aiming to evaluate partition functions of lattice models and the
Kramers-\!Wannier transfer matrix technique \cite{Kramers_1941a, Kramers_1941b} plays a distinguished role.
Within the Kramers-\!Wannier approach, the computation of the model's partition function is converted into an eigenvalue problem for the associated 
transfer matrix. Although the latter might be regarded as an auxiliary tool for this computation, it gained a more fundamental status with the advent of
Baxter's concept of commuting transfer matrices \cite{Baxter_1971}. For instance, in models possessing commuting transfer matrices there usually exist
functional equations characterizing the transfer matrices eigenvalues. The largest eigenvalue is then responsible for the model's
free-energy in the thermodynamical limit. However, it is worth remarking that a more direct description of partition functions in terms of
functional equations has also been put forward in \cite{Stroganov_1979}.

The six-vertex model with domain-wall boundaries is a singular example where the computation of the model's partition function can not
be readily recast as an eigenvalue problem along the lines of Kramers-\!Wannier technique. Nevertheless, this particular model still
admits a description in terms of a spectral problem as shown in \cite{Galleas_Twists}. As far as this variant of the six-vertex model
is concerned, the possibility of describing its partition function in terms of functional equations was firstly discussed in
\cite{Stroganov_2002}. The functional equation presented in \cite{Stroganov_2002} is, however, restricted to values of the model anisotropy 
parameter satisfying a root-of-unity condition. For arbitrary values of the anisotropy parameter an alternative functional equation was 
subsequently derived in \cite{Galleas_2010, Galleas_2013}. The latter originates from the Yang-Baxter algebra and it will be the starting point
of our present analysis.

\paragraph{Functional equation.} Let $Z$ be the partition function of the six-vertex model with domain-wall boundaries
on a $L \times L$ lattice \cite{Korepin_1982}. It is a symmetric function $Z \colon \mathbb{C}^L \to \mathbb{C}$ depending 
on $L$ spectral variables $\lambda_i \in \mathbb{C}$. Strictly speaking $Z$ also depends on a set of $L$ variables $\mu_i \in \mathbb{C}$, 
usually refereed to as inhomogeneity parameters, but here they will be fixed. 
Now write $X^{i,j} \coloneqq \{ \lambda_k \mid i \leq k \leq j \}$ and define the sets
$X^{i,j}_k \coloneqq X^{i,j} \backslash \{ \lambda_k \}$. Using the algebraic-functional method formulated in \cite{Galleas_2013, Galleas_2008},
one can show that the partition function $Z$ satisfies the following functional equation 
\[ 
\label{FZ}
\sum_{i=0}^L M_i \; Z(X^{0,L}_i) = 0 \; ,
\]
with coefficients
\<
\label{cfc}
M_i &\coloneqq& \begin{cases} 
\displaystyle \prod_{j=1}^{L} b(\lambda_0 - \mu_j) - \prod_{j=1}^{L} a(\lambda_0 - \mu_j) \prod_{j=1}^{L} \frac{a(\lambda_j - \lambda_0)}{b(\lambda_j - \lambda_0)}  \qquad \quad \; \mbox{for} \; i = 0 \\ 
\displaystyle \frac{c(\lambda_i - \lambda_0)}{b(\lambda_i - \lambda_0)} \prod_{j=1}^{L} a(\lambda_i - \mu_j) \prod_{\stackrel{j=1}{j \neq i}}^{L} \frac{a(\lambda_j - \lambda_i)}{b(\lambda_j - \lambda_i)} \qquad \qquad \quad \quad \mbox{otherwise} \; . 
\end{cases} \nonumber \\
\>

The functions $a$, $b$ and $c$ in \eqref{cfc} correspond to the statistical weights of the symmetric six-vertex model and here we shall
restrict our discussion to the rational model. In that case we have $a(\lambda) \coloneqq \lambda + \eta$, $b(\lambda) \coloneqq \lambda$
and $c(\lambda) \coloneqq \eta$. We can readily see that $a = b + c$ which corresponds to a quotient of the six-vertex
model algebraic curve $a^2 + b^2 - c^2 = \Delta a b$ with $\Delta = 2$. The parameter $\eta$ is usually refereed to as \emph{semi-classical
parameter} and here it will also be fixed. In this way $c(\lambda)$ can be regarded as a constant.

As a matter of fact Eq. \eqref{FZ} does not consist of a single functional equation for the partition function $Z$. It comprises a total
of $L+1$ equations uncovered through permutations $\lambda_0 \leftrightarrow \lambda_i$. This feature has its roots in the following
properties:
\begin{enumerate}[label=\emph{\roman*}.] 
\item Eq. \eqref{FZ} runs over $L+1$ variables whilst $Z$ depends only on $L$ variables;
\item Analytic solutions are symmetric as demonstrated in \cite{Galleas_2012, Galleas_2013};
\item Although \eqref{FZ} is invariant under permutations $\lambda_i \leftrightarrow \lambda_j$ for $1 \leq i,j \leq L$, the permutation
$\lambda_0 \leftrightarrow \lambda_i$ produces a new equation with the same structure but modified coefficients.
\end{enumerate}

The above remarks pave the way for extending (\ref{FZ}) to the following set of functional relations, namely
\[ 
\label{FZs}
\sum_{i=0}^L M_i^{(n)} \; Z(X^{0,L}_i) = 0 \qquad \qquad n =  1, 2,  \dots , L
\]
with coefficients
\< \label{cfcn}
M_i^{(n)} \coloneqq \begin{cases}
\left. M_n \right|_{\lambda_0 \leftrightarrow \lambda_n} \qquad \mbox{for} \; i = 0 \\
\left. M_0 \right|_{\lambda_0 \leftrightarrow \lambda_n} \qquad \mbox{for} \; i = n \\
\left. M_i \right|_{\lambda_0 \leftrightarrow \lambda_n} \qquad \; \mbox{otherwise} 
\end{cases} \; .
\>

In what follows we shall describe a mechanism rendering a system of first-order partial differential equations 
similar to the classical Knizhnik-Zamolodchikov equation from (\ref{FZ}) and (\ref{FZs}).

\section{Partial differential equations}
\label{sec:KZ}

Several analogies between (\ref{FZ}) and the theory of Knizhnik-Zamolodchikov equations have already been pointed out in \cite{Galleas_proc}. 
However, the different nature of the aforementioned equations seems to prevent a direct identification at first sight. 
Although the existence of a precise matching is not clear at the moment here we intend to shed some light onto possible connections by unveiling
a system of first-order PDEs underlying Eqs. (\ref{FZ}) and (\ref{FZs}).

In order to proceed, we shall firstly restrict our attention to Eq. (\ref{FZ}). In \Secref{sec:FUN} we have already pointed out
the distinguished role played by the variable $\lambda_0$ and we can readily see from (\ref{cfc}) that the coefficients  $M_0$ and
$M_k$ exhibit a simple pole when $\lambda_0 = \lambda_k$. Despite the existence of such poles the limit $\lambda_0 \to \lambda_k$ 
in (\ref{FZ}) is still well defined and by taking this limit we are left with a first-order PDE satisfied by the partition 
function $Z$. For later convenience we also introduce the following extra definition.

\begin{mydef}[Evaluation map] Let $\Lambda_n \coloneqq \mathbb{C} \llbracket x_1^{\pm 1} , x_2^{\pm 1} , \dots , x_n^{\pm 1} \rrbracket$
be the ring of formal Laurent series in $n$ variables over $\mathbb{C}$ and let $\mathfrak{S}_n$ denote the symmetric group on $n$ letters.
In addition to that define $\bar{\Lambda}_n \coloneqq \mathbb{C} \llbracket x_1^{\pm 1} , x_2^{\pm 1} , \dots , x_n^{\pm 1} \rrbracket^{\mathfrak{S}_n}$
as the subset $\bar{\Lambda}_n \subset \Lambda_n$ formed by symmetric Laurent series. Next set $i,j \in \{ 1, 2, \dots , n \}$
and the evaluation map $\mathcal{E} \colon \bar{\Lambda}_n \to \Lambda_{n-1}$ is defined by the plethystic substitution
\[
\label{spec}
\left( \mathcal{E}_{ij} f \right) (x_1, x_2 , \dots , x_n) \coloneqq \left. f(x_1, x_2 , \dots , x_n) \right|_{x_j \mapsto x_i}
\]
for $f \in \bar{\Lambda}_n$ such that the RHS of (\ref{spec}) is well defined.
\end{mydef}
  
\bigskip

\begin{theorem} \label{KZt}
The partition function $Z$ satisfies the system of PDEs
\[
\label{KZE}
c \; \partial_i Z = \bigg( \sum_{\substack{j=1 \\ j \neq i}}^L \gen{\Omega}_{ij} + h_i \bigg) Z  \; ,
\]
where $\partial_i \coloneqq \frac{\partial}{\partial \lambda_i}$ and $i \in \{1, 2, \dots, L \}$. The terms
$\gen{\Omega}_{ij}$ and $h_i$ are in their turn defined by
\<
\label{omegah}
\gen{\Omega}_{ij} &\coloneqq& \frac{c(\lambda_j - \lambda_i)}{a(\lambda_j - \lambda_i)} \frac{a(\lambda_i - \lambda_j)}{b(\lambda_i - \lambda_j)} \prod_{k=1}^{L} \frac{a(\lambda_j - \mu_k)}{a(\lambda_i - \mu_k)} \prod_{\substack{k=1 \\ k \neq i,j}}^{L}  \frac{b(\lambda_k - \lambda_i)}{a(\lambda_k - \lambda_i)} \frac{a(\lambda_k - \lambda_j)}{b(\lambda_k - \lambda_j)} \mathcal{E}_{ij} \nonumber \\
h_i &\coloneqq& \prod_{k=1}^{L} \frac{b(\lambda_i - \mu_k)}{a(\lambda_i - \mu_k)} \prod_{\substack{k=1 \\ k \neq i}}^{L}  \frac{b(\lambda_k - \lambda_i)}{a(\lambda_k - \lambda_i)} + \sum_{k=1}^L \frac{c(\lambda_i - \mu_k)}{a(\lambda_i - \mu_k)} + \sum_{\substack{k=1 \\ k \neq i}}^L \frac{c(\lambda_k - \lambda_i)}{b(\lambda_k - \lambda_i)} \frac{c(\lambda_k - \lambda_i)}{a(\lambda_k - \lambda_i)} -1 \; . \nonumber \\
\>
\end{theorem}
\begin{proof}
Set $\lambda_0 = \lambda_i + \alpha$ in \eqref{FZ} and expand it in power series in $\alpha$. By doing so, and taking into
account that $Z$ is a symmetric function, we are left with a Taylor series containing only positive powers of $\alpha$. Then 
\eqref{KZE} is obtained by setting $\alpha = 0$. Alternatively, one can also take directly the limit $\lambda_0 \to \lambda_i$ 
of \eqref{FZ} using L'Hopital's rule.
\end{proof}

The structure of (\ref{KZE}) resembles that of a classical Knizhnik-Zamolodchikov equation and we refer the reader to \cite{Varchenko_book, Etingof_book}
for more details on such equations and their generalizations. Moreover, Eq. (\ref{FZ}) is not the only functional equation originated from the algebraic-functional framework
as discussed in \Secref{sec:FUN}. It turns out that we can associate a system of first-order PDEs to each one of the equations (\ref{FZs}). 
Those equations exhibit the same structure of (\ref{KZE}) and they read as follows.

\bigskip
\begin{theorem} \label{KZf}
For $i,n \in \{ 1, 2, \dots , L \}$ the partition function $Z$ satisfies the following family of partial differential equations, 
\[
\label{KZF}
c \; \partial_i Z = \bigg( \sum_{\substack{j=1 \\ j \neq i}}^L \gen{\Omega}_{ij}^{(n)} + h_i^{(n)}  \bigg) Z  \; .
\]
The terms $\gen{\Omega}_{ij}^{(n)}$ and $h_i^{(n)}$ in \eqref{KZF} are explicitly given by
\<
\label{omegahn}
\gen{\Omega}_{ij}^{(n)} &\coloneqq& \frac{a(\lambda_i - \lambda_j)}{b(\lambda_j - \lambda_i)} \frac{a(\lambda_n - \lambda_i)}{a(\lambda_n - \lambda_j)} \prod_{k=1}^L \frac{a(\lambda_j - \mu_k)}{a(\lambda_i - \mu_k)} \prod_{\substack{k=1 \\ k \neq i}}^L \frac{b(\lambda_k - \lambda_i)}{a(\lambda_k - \lambda_i)} \prod_{\substack{k=1 \\ k \neq j}}^L \frac{a(\lambda_k - \lambda_j)}{b(\lambda_k - \lambda_j)}  \mathcal{E}_{ij} \qquad j \neq n \nonumber \\
\gen{\Omega}_{in}^{(n)} &\coloneqq& c^{-1} \prod_{k=1}^L a^{-1} (\lambda_i - \mu_k) \prod_{\substack{k=1 \\ k \neq i,n}}^L \frac{b(\lambda_k - \lambda_i)}{a(\lambda_k - \lambda_i)} \nonumber \\
&& \left[ b(\lambda_n - \lambda_i) \prod_{k=1}^L b(\lambda_n - \mu_k) + \frac{a^2 (\lambda_i - \lambda_n)}{b(\lambda_i - \lambda_n)} \prod_{k=1}^L a(\lambda_n - \mu_k) \prod_{\substack{k=1 \\ k \neq i,n}}^L \frac{a(\lambda_k - \lambda_n)}{b(\lambda_k - \lambda_n)} \right] \mathcal{E}_{in} \nonumber \\
h_i^{(n)} &\coloneqq& \sum_{\substack{j=1 \\ j \neq i,n}}^L \frac{c(\lambda_j - \lambda_i)}{b(\lambda_j - \lambda_i)} \frac{c(\lambda_j - \lambda_i)}{a(\lambda_j - \lambda_i)} + \sum_{j=1}^L \frac{c(\lambda_i - \mu_j)}{a(\lambda_i - \mu_j)}  -  \frac{a(\lambda_i - \lambda_n)}{b(\lambda_i - \lambda_n)} - 1 \; . 
\>
\end{theorem}
\begin{proof}
Same proof as for Theorem \ref{KZt}.
\end{proof}

\begin{remark}
Eq. (\ref{KZF}) for $i=n$ coincides with Eq. (\ref{KZE}). Thus we have a total of $L$ systems of PDEs and not $L+1$ as one might have expected.
We shall restrict ourselves only to the system (\ref{KZF}) through our analysis in order to avoid over counting.
\end{remark}

\section{Solution}
\label{sec:SOL}

The resolution of the functional equation (\ref{FZ}) has been already described in \cite{Galleas_2013}. In fact, the method devised
in \cite{Galleas_2013} seems to apply to several functional equations with structure resembling (\ref{FZ}). In this way we have 
obtained multiple contour integral representations for a variety of partition functions and scalar products associated to
six-vertex models through this approach \cite{Galleas_2012, Galleas_2013, Galleas_SCP, Galleas_Lamers_2014, Galleas_openSCP}. 
Despite the sound results, the applicability of this method still seems to rely on special properties of the desired solution which
might not exist for certain models. For instance, the method devised in \cite{Galleas_2012, Galleas_2013} depends strongly on the existence of
the so called \emph{special zeroes} which, in case they exist, might still be difficult to locate. This is precisely the point where the formulation 
(\ref{KZF}) seems to be more suitable than (\ref{FZ}). In what follows we shall demonstrate how the solution of (\ref{KZF}) naturally emerges from its structure without requiring strong properties of the partition 
function $Z$.

Theorem \ref{KZf} gives us a total of $L$ systems of first-order PDEs and we can readily notice the presence of $L-1$ terms of the form
$\gen{\Omega}_{ij}^{(n)} Z$ in the RHS of (\ref{KZF}). Our equations can be fortunately solved for those terms and in order to see that we write
\[ \label{som}
\gen{\Omega}_{ij}^{(n)} Z \eqqcolon \omega_{ij}^{(n)} \mathcal{E}_{ij}(Z)
\]
in such a way that $\omega_{ij}^{(n)}$ can be directly read off from \eqref{omegahn}.

Next consider the subset of \eqref{KZF} formed by $n \in \{ 1, 2 , \dots , i-1 , i+1 , \dots , L \}$. We can conveniently rewrite 
the associated $L-1$ system of equations as 
\< \label{subKZ}
\sum_{\substack{j=1 \\ j \neq i}}^L \omega_{ij}^{(n)} \mathcal{E}_{ij}(Z) = \left(c \partial_i - h_i^{(n)} \right) Z  \; , 
\>
in order to make the dependence with the $L-1$ terms $\mathcal{E}_{ij}(Z)$ more apparent. In this way we can readily apply Cramer's rule and by doing
so we find
\[
\label{sji}
\mathcal{E}_{ij} (Z) = \frac{\mbox{det}(H_{ij})}{\mbox{det}(W_i)} \partial_i Z - \frac{\mbox{det}(\bar{H}_{ij})}{\mbox{det}(W_i)} Z \; ,
\]
with matrices $W_i$, $H_{ij}$ and $\bar{H}_{ij}$ defined as
\[ \label{WI}
W_i \coloneqq \left( \begin{matrix}
\omega_{i 1}^{(1)} & \hdots & \omega_{i i-1}^{(1)} & \omega_{i i+1}^{(1)} & \hdots & \omega_{i  L}^{(1)} \\
\vdots & \ddots & \vdots & \vdots & \ddots & \vdots \\
\omega_{i 1}^{(i-1)} & \hdots & \omega_{i i-1}^{(i-1)} & \omega_{i i+1}^{(i-1)} & \hdots & \omega_{i L}^{(i-1)} \\
\omega_{i 1}^{(i+1)} & \hdots & \omega_{i i-1}^{(i+1)} & \omega_{i i+1}^{(i+1)} & \hdots & \omega_{i L}^{(i+1)} \\
\vdots & \ddots & \vdots & \vdots & \ddots & \vdots \\
\omega_{i 1}^{(L)} & \hdots & \omega_{i i-1}^{(L)} & \omega_{i i+1}^{(L)} & \hdots & \omega_{i L}^{(L)} \end{matrix} \right) \; ,
\]
\[ \label{bHIJ}
H_{ij} \coloneqq \left( \begin{matrix}
\omega_{i 1}^{(1)} & \hdots & \omega_{i j-1}^{(1)} & c & \omega_{i j+1}^{(1)} & \hdots & \omega_{i  L}^{(1)} \\
\vdots & \ddots & \vdots & \vdots & \vdots & \ddots & \vdots \\
\omega_{i 1}^{(i-1)} & \hdots & \omega_{i j-1}^{(i-1)} & c &  \omega_{i j+1}^{(i-1)} & \hdots & \omega_{i L}^{(i-1)} \\
\omega_{i 1}^{(i+1)} & \hdots & \omega_{i j-1}^{(i+1)} & c &  \omega_{i j+1}^{(i+1)} & \hdots & \omega_{i L}^{(i+1)} \\
\vdots & \ddots & \vdots & \vdots & \vdots & \ddots & \vdots \\
\omega_{i 1}^{(L)} & \hdots & \omega_{i j-1}^{(L)} & c & \omega_{i j+1}^{(L)} & \hdots & \omega_{i  L}^{(L)} \end{matrix} \right) \; ,
\]
and
\[ \label{HIJ}
\bar{H}_{ij} \coloneqq \left( \begin{matrix}
\omega_{i 1}^{(1)} & \hdots & \omega_{i j-1}^{(1)} & h_i^{(1)} & \omega_{i j+1}^{(1)} & \hdots & \omega_{i  L}^{(1)} \\
\vdots & \ddots & \vdots & \vdots & \vdots & \ddots & \vdots \\
\omega_{i 1}^{(i-1)} & \hdots & \omega_{i j-1}^{(i-1)} & h_i^{(i-1)} &  \omega_{i j+1}^{(i-1)} & \hdots & \omega_{i L}^{(i-1)} \\
\omega_{i 1}^{(i+1)} & \hdots & \omega_{i j-1}^{(i+1)} & h_i^{(i+1)} &  \omega_{i j+1}^{(i+1)} & \hdots & \omega_{i L}^{(i+1)} \\
\vdots & \ddots & \vdots & \vdots & \vdots &  \ddots & \vdots \\
\omega_{i 1}^{(L)} & \hdots & \omega_{i j-1}^{(L)} & h_i^{(L)} & \omega_{i j+1}^{(L)} & \hdots & \omega_{i  L}^{(L)} \end{matrix} \right) \; .
\]
As far as formulae (\ref{bHIJ}) and (\ref{HIJ}) are concerned it is worth stressing that columns containing terms $\omega_{i i}^{(n)}$ are absent 
from the definition of the matrices $H_{ij}$ and $\bar{H}_{ij}$. Therefore formula (\ref{sji}) expresses each single term 
$\mathcal{E}_{ij} (Z)$ as a linear combination of $Z$ and the partial derivative $\partial_i Z$. This particular type of relation will be the main ingredient for establishing the
following theorem.

\begin{theorem} \label{DZ}
The partition function $Z$ satisfies the system of equations
\[
\label{fuchs}
\partial_i Z = \frac{\gen{det}(\bar{\gen{K}}_i)}{\gen{det}(\gen{K}_i)} Z  \; ,
\]
where 
\< \label{AI}
\bar{\gen{K}}_i &\coloneqq& \left( \begin{matrix}
\omega_{i 1}^{(1)} & \hdots & \omega_{i i-1}^{(1)} & h_i^{(1)}  & \omega_{i i+1}^{(1)} & \hdots & \omega_{i  L}^{(1)} \\
\vdots & \ddots & \vdots & \vdots & \vdots  & \ddots & \vdots \\ 
\omega_{i 1}^{(L)} & \hdots & \omega_{i i-1}^{(L)} & h_i^{(L)}  & \omega_{i i+1}^{(L)} & \hdots & \omega_{i L}^{(L)} \end{matrix} \right)  
\>
and
\< \label{AII}
\gen{K}_i &\coloneqq& \left( \begin{matrix}
\omega_{i 1}^{(1)} & \hdots & \omega_{i i-1}^{(1)} & c  & \omega_{i i+1}^{(1)} & \hdots & \omega_{i  L}^{(1)} \\
\vdots & \ddots & \vdots & \vdots & \vdots  & \ddots & \vdots \\ 
\omega_{i 1}^{(L)} & \hdots & \omega_{i i-1}^{(L)} & c  & \omega_{i i+1}^{(L)} & \hdots & \omega_{i L}^{(L)} \end{matrix} \right)  \; .
\>
\end{theorem}
\begin{proof}
Substitute (\ref{sji}) into (\ref{KZF}) with $n=i$. This procedure leaves us with the equation $\partial_i Z = B_i Z $ where
\[
B_i \coloneqq  \frac{\displaystyle h_i^{(i)} \mbox{det}(W_i) - \sum_{\substack{j=1 \\ j \neq i }}^L \omega_{ij}^{(i)} \mbox{det}(\bar{H}_{ij})}
{c \; \mbox{det}(W_i) - \displaystyle \sum_{\substack{j=1 \\ j \neq i }}^L \omega_{ij}^{(i)} \mbox{det}(H_{ij}) } \; .
\]
Next we use determinant expansion by minors to verify that the denominator of $B_i$ corresponds to $\mbox{det}(\gen{K}_i)$
while the numerator equals $\mbox{det}(\bar{\gen{K}}_i)$.
\end{proof}

\paragraph{Simplifications.} The matrices $\bar{\gen{K}}_i$ and $\gen{K}_i$ respectively defined by (\ref{AI}) and (\ref{AII})
are given in terms of functions $\omega_{ij}^{(n)}$ and $h_i^{(n)}$ obtained through the relations \eqref{omegahn} and \eqref{som}. 
As a matter of fact we are interested only in the determinants of these matrices and that is still liable to simplifications.
In this way it is convenient to rewrite the required determinants as
\<
\label{KtoY}
\mbox{det} (\gen{K}_i) &=& \prod_{k=1}^L a^{1-L} (\lambda_i - \mu_k) \prod_{\substack{k=1 \\ k \neq i}}^L \frac{b^{L-1} (\lambda_k - \lambda_i) }{a^{L-1} (\lambda_k - \lambda_i)} \mbox{det}(\gen{Y}_i)  \nonumber \\
\mbox{det} (\bar{\gen{K}}_i) &=& \prod_{k=1}^L a^{1-L} (\lambda_i - \mu_k) \prod_{\substack{k=1 \\ k \neq i}}^L \frac{b^{L-1} (\lambda_k - \lambda_i) }{a^{L-1} (\lambda_k - \lambda_i)} \mbox{det}(\bar{\gen{Y}}_i) 
\> 
with matrices $\gen{Y}_i$ and $\bar{\gen{Y}}_i$ obtained respectively from $\gen{K}_i$ and $\bar{\gen{K}}_i$ under the replacement
$\omega_{ij}^{(n)} \mapsto \bar{\omega}_{ij}^{(n)}$. The functions $\bar{\omega}_{ij}^{(n)}$ are then given by
\<
\label{BW}
\bar{\omega}_{ij}^{(n)} &\coloneqq& \frac{a(\lambda_i - \lambda_j)}{b(\lambda_j - \lambda_i)} \frac{a(\lambda_n - \lambda_i)}{a(\lambda_n - \lambda_j)} \prod_{k=1}^L a(\lambda_j - \mu_k) \prod_{\substack{k=1 \\ k \neq j}}^L \frac{a(\lambda_k - \lambda_j)}{b(\lambda_k - \lambda_j)} \qquad \quad \qquad \qquad \mbox{for} \; j \neq n, \nonumber \\
\bar{\omega}_{in}^{(n)} &\coloneqq& \frac{a(\lambda_n - \lambda_i)}{c(\lambda_n - \lambda_i)} \left[ \prod_{k=1}^L b(\lambda_n - \mu_k) - \left( \frac{a (\lambda_i - \lambda_n)}{b(\lambda_i - \lambda_n)} \right)^2 \prod_{k=1}^L a(\lambda_n - \mu_k) \prod_{\substack{k=1 \\ k \neq i,n}}^L \frac{a(\lambda_k - \lambda_n)}{b(\lambda_k - \lambda_n)} \right]  \; . \nonumber \\
\>
The explicit evaluation of $\mbox{det}(\gen{Y}_i)$ and $\mbox{det}(\bar{\gen{Y}}_i)$ for small values of $L$ reveals that
they are in fact multivariate polynomials and this seems to be the case for arbitrary values of $L$. Moreover, although we shall not present a proof
here, $\mbox{det}(\gen{Y}_i)$ is a multivariate polynomial of degree $L-1$ in each variable separately while $\mbox{det}(\bar{\gen{Y}}_i)$
is of degree $L-2$. The inspection of these polynomials for small values of $L$ also indicates that they share no common zeroes for generic
values of the variables $\lambda_j$. Consequently, the ratio $\mbox{det}(\bar{\gen{Y}}_i) / \mbox{det}(\gen{Y}_i)$
is a rational function preserving the polynomial degree of both numerator and denominator.

\paragraph{The partition function $Z$.} Theorem \ref{DZ} states that the partition function $Z$ satisfies a
simple system of linear first-order PDEs whose coefficients are expressed as determinants.
Considering \eqref{fuchs} and \eqref{KtoY} we then have
\[
\label{fuchsY}
\partial_i Z = \frac{\mbox{det}(\bar{\gen{Y}}_i)}{\mbox{det}(\gen{Y}_i)} Z  \; .
\]
Eq. (\ref{fuchsY}) has interesting consequences for functions $Z \in C^1 (\mathbb{C}^L)$ and the solution
can be obtained as follows.

\begin{corollary}
The solution of (\ref{fuchsY}) with asymptotic behavior 
\[ \label{asymp}
Z \sim 2^{-L(L-1)} \; c^L \; L! \; \prod_{j=1}^L \lambda_j^{L-1} \qquad \mbox{as} \quad \lambda_j \to \infty \; ,
\]
admits the family of representations
\[
\label{ZY}
Z = c^{L-1} \; \gen{det}(\gen{Y}_i) \; .
\]
\end{corollary}
\begin{proof}
Restricting our attention to functions $Z \in C^1 (\mathbb{C}^L)$ we can conclude that the residues of
the RHS of (\ref{fuchsY}) at the zeroes of $\mbox{det}(\gen{Y}_i)$ must vanish. In this way, $Z$ needs to have
the same zeroes as $\mbox{det}(\gen{Y}_i)$ since $\mbox{det}(\bar{\gen{Y}}_i)$ does not share zeroes with $\mbox{det}(\gen{Y}_i)$
for generic values of $\lambda_j$. We lack a proof for the latter statement but we have verified the validity of this
property for several values of $L$. As previously mentioned, $\mbox{det}(\gen{Y}_i)$ is a polynomial of degree $L-1$
in each variable $\lambda_j$ separately and consequently it has $L-1$ zeroes with respect to a given variable $\lambda_j$.
Therefore, $Z$ needs to have at least $L-1$ zeroes with respect to each variable $\lambda_j$. Now let 
$\mathcal{V}_j \coloneqq \mathbb{C}[\lambda_1^{\pm 1}, \dots , \lambda_{j-1}^{\pm 1} , \lambda_{j+1}^{\pm 1} , \dots , \lambda_{L}^{\pm 1} ]$
and define the set of zeroes $\mathcal{L}_j^{Z} \coloneqq \{ \lambda \in \mathcal{V}_j \mid \left. Z \right|_{\lambda_j = \lambda} = 0  \}$ and
$\mathcal{L}_j^{\gen{Y}_i} \coloneqq \{ \lambda \in \mathcal{V}_j \mid \left. \mbox{det}(\gen{Y}_i) \right|_{\lambda_j = \lambda} = 0  \}$.
Thus $\mbox{card} (\mathcal{L}_j^{\gen{Y}_i}) = L-1$ while $\mbox{card} (\mathcal{L}_j^{Z}) \geq L-1$.
Next assume $\exists \zeta \in \mathcal{L}_j^{Z} \backslash \mathcal{L}_j^{\gen{Y}_i}$ and suppose $\mathcal{C}$ is
an integration contour enclosing solely $\zeta$. Under these assumptions Eq. (\ref{fuchsY}) can be integrated as follows,
\[ \label{pp}
\oint_{\mathcal{C}} \frac{\partial_i Z}{Z} \dd \lambda_j = \oint_{\mathcal{C}} \frac{\mbox{det}(\bar{\gen{Y}}_i)}{\mbox{det}(\gen{Y}_i)} \dd \lambda_j  \; .
\]
The RHS of (\ref{pp}) vanishes while the LHS is different from zero if such element $\zeta$ exists.
Hence we can conclude that $\mathcal{L}_j^{Z} = \mathcal{L}_j^{\gen{Y}_i}$ by contradiction and consequently $Z = r_i \; \gen{det}(\gen{Y}_i)$.
The constant $r_i$ is then fixed by the asymptotic behavior (\ref{asymp}) derived in \cite{Galleas_2010}. By doing so we find
$r_i = c^{L-1}$.
\end{proof}

\begin{remark}
The RHS of \eqref{ZY} depends on an index $i \in \{1, 2, \dots , L \}$ and it consequently comprises a total of $L$ representations
for the partition function $Z$. Although we have not presented a rigorous proof here, the explicit evaluation of $\gen{det}(\gen{Y}_i)$ and 
$\gen{det}(\gen{Y}_j)$ for $i\neq j$ and several values of $L$ confirms they result in the same partition function $Z$. Moreover, 
$\gen{det}(\gen{Y}_i)$ and $\gen{det}(\gen{Y}_j)$ do not seem to be related by any obvious determinant preserving transformation and 
this fact suggests they indeed consist of independent representations.
\end{remark}

\section{Concluding remarks}
\label{sec:CONCLUSION}

In this work we have studied the partition function of the six-vertex model with domain-wall boundaries
from the perspective of partial differential equations. We have shown that a previously obtained functional 
equation describing the model's partition function can be reduced to a system of first-order differential equations
resembling the Knizhnik-Zamolodchikov equation. In fact, we have obtained a whole family of systems of equations which can be 
solved by elementary methods. The solution is then obtained as a single determinant whose partial 
homogeneous limit $\mu_j \mapsto \mu$ can be obtained trivially. Interestingly, the procedure described here
yields a whole family of representations which do not seem to be related in any obvious way.

The starting point of our analysis is the functional equation \eqref{FZ} which can also be solved in terms of a multiple contour integral
as shown in \cite{Galleas_2013, Galleas_proc}. However, the resolution of \eqref{FZ} along the lines described in \cite{Galleas_2013}
requires the use of additional properties of the desired solution, for example, they exhibit a particular polynomial structure. 
Furthermore, the implementation of the approach described in \cite{Galleas_2013} depends strongly on the location of the so-called 
\emph{special zeroes} which might not even exist for certain models. This seems to be the case for models based on non-compact symmetries 
such as the non-linear Schr\"odinger model \cite{Sklyanin_1979, Korepin_book}. 

As for the method we are employing here to solve the system of PDEs (\ref{KZF}), some remarks are also important. For instance, one of the arguments
we have used for obtaining \eqref{ZY} is that $\gen{det}(\gen{Y}_i)$ and $\gen{det}(\bar{\gen{Y}}_i)$ are polynomials of particular
degrees sharing no common zeroes. A rigorous proof of both arguments is missing but we have verified the validity of these assumptions
for several lattice lengths. Nevertheless, under these assumptions we find that the full characterization of the desired partition function
does not require much more input than what is already contained in our system of equations. 

The partition function of the six-vertex model with domain-wall boundary conditions has been already extensively studied in the literature
and many of its properties are well known. For example, one of the key results in the literature is the determinant representation
obtained in \cite{Izergin_1987} which has played an important role for further studies of correlation functions of quantum integrable systems
\cite{Korepin_book}. In spite of that, the structure of the functional equation (\ref{FZ}) is shared by several other quantities as shown
in \cite{Galleas_2013, Galleas_SCP, Galleas_Lamers_2014, Galleas_openSCP, Galleas_Lamers_2015} and one can expect that our present analysis
can be extended to those cases as well. More importantly, this analysis might be of particular relevance for the elliptic SOS model 
considered in \cite{Galleas_2013} and \cite{Rosengren_2008} for which a single determinant representation for the model's partition function
is unknown to date.

\section{Acknowledgements}
\label{sec:ACK}

The work of W.G.\ is supported by the German Science Foundation (\textsc{DFG}) under the Collaborative Research Center 
(\textsc{sfb}) 676, \textit{Particles, Strings and the Early Universe}.

%
%
%

\bibliographystyle{hunsrt}
\bibliography{references}

\end{document}